\documentclass[12pt]{amsart}
\usepackage{}

\usepackage{amsmath}
\usepackage{amsfonts}
\usepackage{amssymb}
\usepackage[all]{xy}           

\usepackage{bbding}
\usepackage{txfonts}
\usepackage{amscd}

\usepackage[shortlabels]{enumitem}
\usepackage{ifpdf}
\ifpdf
  \usepackage[colorlinks,final,backref=page,hyperindex]{hyperref}
\else
  \usepackage[colorlinks,final,backref=page,hyperindex,hypertex]{hyperref}
\fi
\usepackage{tikz}
\usepackage[active]{srcltx}

\makeatletter

\newtheorem{thm}{Theorem}[section]

\newtheorem{cor}[thm]{Corollary}
\newtheorem{pro}[thm]{Proposition}
\newtheorem{ex}[thm]{Example}
\newtheorem{rmk}[thm]{Remark}
\newtheorem{defi}[thm]{Definition}

\newcommand {\emptycomment}[1]{}

\newcommand{\lon }{\,\rightarrow\,}
\newcommand{\be }{\begin{equation}}
\newcommand{\ee }{\end{equation}}

\newcommand{\g}{\mathfrak g}
\newcommand{\h}{\mathfrak h}

\newcommand{\huaG}{\mathcal{G}}

\newcommand{\huaO}{{\mathcal{O}}}

\newcommand{\InnDer}{\mathrm{InnDer}}
\newcommand{\Orb}{\mathrm{Orb}}
\newcommand{\dM}{\mathrm{d}}
\newcommand{\NR}{\mathrm{NR}}

\newcommand{\Courant}[1]{\left\llbracket  #1\right\rrbracket }

\newcommand{\Id}{{\rm{Id}}}

\newcommand{\br}[1]{   [ \cdot,    \cdot  ]   }

\newcommand{\CE}{\mathsf{CE}}

\newcommand{\Hom}{\mathrm{Hom}}
\newcommand{\InnAut}{\mathrm{InnAut}}

\newcommand{\gl}{\mathfrak {gl}}

\newcommand{\ad}{\mathrm{ad}}

\begin{document}

\title[Deformations of modified $r$-matrices]{Deformations of modified $r$-matrices and cohomologies of related algebraic
 structures}

\author{Jun Jiang}
\address{Department of Mathematics, Jilin University, Changchun 130012, Jilin, China}
\email{junjiang@jlu.edu.cn}

\author{Yunhe Sheng}
\address{Department of Mathematics, Jilin University, Changchun 130012, Jilin, China}
\email{shengyh@jlu.edu.cn}


\begin{abstract}
Modified $r$-matrices are solutions of the modified classical Yang-Baxter equation, introduced by Semenov-Tian-Shansky, and play important roles in mathematical physics. In this paper, first we introduce a cohomology theory for modified $r$-matrices. Then we study three kinds of deformations of modified $r$-matrices using the established cohomology theory, including algebraic deformations, geometric deformations and linear deformations. We give the differential graded Lie algebra that governs algebraic deformations of modified $r$-matrices.  For geometric deformations, we prove the rigidity theorem and study when is a neighborhood of a modified $r$-matrix smooth in
the space of all modified $r$-matrix structures. In the study of trivial linear deformations, we introduce the notion of a Nijenhuis element for a modified $r$-matrix. Finally, applications are given to study deformations of complement of the diagonal Lie algebra and compatible Poisson structures.
\end{abstract}

\renewcommand{\thefootnote}{}
\footnotetext{2020 Mathematics Subject Classification. 17B37, 17B38, 17B56}

\keywords{modified classical Yang-Baxter equation, modified $r$-matrix, cohomology, deformation}

\maketitle

\tableofcontents

\allowdisplaybreaks


\section{Introduction}

In the seminal work \cite{STS}, Semenov-Tian-Shansky showed that solutions of the
modified classical Yang-Baxter equation, which we call modified $r$-matrices in this paper, play an important role in
studying solutions of Lax equations \cite{RS1, STS,STS2}. Furthermore,  modified $r$-matrices are intimately related to particular factorization problems in
the corresponding Lie algebras and Lie groups. This factorization problem was considered by Reshetikhin and  Semenov-Tian-Shansky in the framework of the enveloping algebra of a Lie algebra with a modified $r$-matrix to study quantum integrable systems \cite{RS88}. Any modified $r$-matrix induces a post-Lie algebra \cite{BGN}, and a factorization theorem for group-like elements of the completion of the Lie enveloping algebra
of a post-Lie algebra was established by   Ebrahimi-Fard,  Mencattini and   Munthe-Kaas in \cite{FKK,FM}. Recently, the global factorization theorem for a Rota-Baxter Lie group was given in \cite{GLS}. Moreover, modified $r$-matrices are also useful for the construction of flat metrics and Frobenius manifolds \cite{Sza}, and compatible Poisson structures \cite{Li}. Note that in the associative algebra context, such objects are called modified Rota-Baxter algebras by Zhang,   Gao and  Guo \cite{ZGG1,ZGG2}.

A classical approach to study a mathematical structure is to associate to it invariants. Among these, cohomology theories occupy a  central position as they enable for example to control deformations or extension problems.
Note that the cohomology theory for a skew-symmetric classical $r$-matrix was studied  in \cite{TBGS} under the general framework of relative Rota-Baxter operators (also called $\huaO$-operators \cite{Ku}). The first purpose of this paper is to study the cohomology theory for a modified $r$-matrix. In  \cite{STS}, Semenov-Tian-Shansky showed that a modified $r$-matrix $R:\g\to\g$ on a Lie algebra $(\g,[\cdot,\cdot]_\g)$ induces a new Lie algebra $\g_R$ in which the Lie bracket $[\cdot,\cdot]_R$ is given by
$$
[x,y]_R=[R(x),y]_\g+[x,R(y)]_\g,\quad \forall x, y\in \g.
$$
In \cite{Bor}, Bordemann showed that the induced Lie algebra $\g_R$ represents on $\g$. We use the corresponding Chevalley-Eilenberg cohomology \cite{Ch-Ei} of the Lie algebra $\g_R$ with coefficients in $\g$ to define the cohomology of the modified $r$-matrix $R$. It is well known that there is a one-to-one correspondence between modified $r$-matrix $R$ and Rota-Baxter operator $B$ of weight 1 via the relation $R=\Id+2B$. The cohomology theory of the latter was given in \cite{JSZ} and the Van Est type theorem was established. We also show that the cohomology of the modified $r$-matrix $R=\Id+2B$ and the cohomology of the Rota-Baxter operator $B$ are isomorphic.

The concept of a formal deformation of an algebraic structure began with the seminal
work of Gerstenhaber~\cite{Ge0,Ge} for associative
algebras. Nijenhuis and Richardson   extended this study to Lie algebras
~\cite{NR,NR2}.  There is a well known slogan, often attributed to Deligne, Drinfeld and Kontsevich:
every reasonable deformation theory is controlled by a differential graded Lie
algebra, determined up to quasi-isomorphism. This slogan has been made into a rigorous
theorem by Lurie and Pridham \cite{Lu,Pr}. 
It is also meaningful to deform {\em maps} compatible with given algebraic structures. Recently, the deformation theory of morphisms was   developed in \cite{Borisov,Fregier-Zambon-1,Fregier-Zambon-2},  the deformation theories of $\huaO$-operators on Lie algebras and associative algebras were developed in \cite{TBGS,Das}. The second purpose of the paper is to study deformation theories of modified $r$-matrices. We study three kinds of deformations of a modified $r$-matrix $R$:
\begin{itemize}
  \item (algebraic deformations) first we consider algebraic deformation $R+R'$ for certain linear map $R'$, and show that this kind of deformations are governed by a differential graded Lie algebra. This fulfill the general slogan for the deformation theory proposed  by Deligne, Drinfeld and Kontsevich;

  \item (geometric deformations) then we consider smooth geometric deformation $R_t$ such that $R_0=R$ using the approach developed by Crainic,   Schatz and  Struchiner in \cite{CSS}. We show that the tangent space  $T_R\Orb_R$ of the orbit $\Orb_R$ is the space of $2$-coboundaries $B^{2}(R)$. Consequently, the condition $H^{2}(R)=0$ will imply certain rigidity theorem, and the condition $H^{3}(R)=0$ will imply the space of modified $r$-matrices on the Lie algebra $\g$ is a manifold in a neighborhood of $R$. We also give the necessary and sufficient condition on a 2-cocycle giving a geometric deformation using the Kuranishi map;

      \item (linear deformation) next we study linear deformation $R+t\hat{R}$. In particular, trivial linear deformations leads to the concept of Nijenhuis elements for a modified $r$-matrix. If $x\in\g$ is a Nijenhuis element, then $\ad_x$ is a Nijenhuis operator on the Lie algebra $\g_R$.
\end{itemize}
Note that certain particular deformation of classical $r$-matrices are considered in \cite{SzaB} in the study of integrable infinite-dimensional
systems.

The papers is organized as follows. In Section \ref{sec:coh}, we define the cohomology of a modified $r$-matrix $R$ using the Chevalley-Eilenberg cohomology   of the Lie algebra $\g_R$ with coefficients in $\g$. In Section \ref{sec:alg}, we construct a differential graded Lie algebra that governs algebraic deformations of a modified $r$-matrix. In Section \ref{sec:geo}, we study geometric deformations of a modified $r$-matrix. In Section \ref{sec:lin}, we study linear deformations of a modified $r$-matrix. In Section \ref{sec:app}, we  study   deformations of complement of the diagonal Lie algebra and compatible Poisson structures as applications.


\section{Cohomologies of modified $r$-matrices}\label{sec:coh}

In this section, we establish  the cohomology theory of a modified $r$-matrix $R$ using the Chevalley-Eilenberg cohomology  of the Lie algebra $\g_R$ with coefficients in $\g$.

\begin{defi}\rm(\cite{STS})
Let $(\g, [\cdot,\cdot]_\g)$ be a Lie algebra. A linear map $R:\g\lon\g$ is called a {\bf modified $r$-matrix} if it is a solution of the following {\bf modified classical Yang-Baxter equation}:
\begin{equation}\label{assmd}
[R(x), R(y)]_\g=R([R(x),y]_\g+[x,R(y)]_\g)-[x,y]_\g, \quad \forall x,y\in\g.
\end{equation}
\end{defi}

\begin{defi}
Let $R$ and $R'$ be modified $r$-matrices on a Lie algebra $(\g,[\cdot,\cdot]_\g)$. A {\bf homomorphism} from $R$ to $R'$ is a Lie algebra homomorphism $\varphi:\g\lon\g$ such that
\begin{equation*}
\varphi\circ R=R'\circ\varphi.
\end{equation*}
\end{defi}

\begin{rmk}
  The notion of a modified Rota-Baxter operator of weight $-1$ on an associative algebra was introduced in \cite{EF}. More precisely, it is a linear map $P: A\lon A$ on an associative algebra $(A, \cdot_A)$ satisfying
\begin{equation*}
P(u)\cdot_{A}P(v)=P(P(u)\cdot_{A}v+u\cdot_{A}P(v))-u\cdot_{A}v, \quad \forall u,v\in A.
\end{equation*}
It is straightforward to see that if a linear map $P: A\lon A$ is a modified Rota-Baxter operator of weight $-1$ on an associative algebra $(A, \cdot_A)$,  then $P$ is a modified $r$-matrix on the Lie algebra $(A, [\cdot,\cdot]_A)$, where $[\cdot,\cdot]_A$ is the commutator Lie bracket.

\end{rmk}

 \begin{rmk}\label{thmex1}
 Let $R:\g\lon\g$  be a linear map on a Lie algebra $(\g, [\cdot,\cdot]_\g)$. Under the condition $R^2=\Id$, the following structures are equivalent:
   \begin{itemize}
     \item $R$ is a  modified $r$-matrix;

     \item $R$ is a Nijenhuis operator;

     \item $R$ is a product structure;

     \item There is a vector space direct sum decomposition $\g=\g_1\oplus\g_2$ of $\g$ into subalgebras $\g_1$ and $\g_2$   such that $R$ is given by
      \begin{equation*}
      R(x,u)= (x,-u),\quad \forall x\in\g_1, u\in\g_2.
      \end{equation*}
   \end{itemize}
 \end{rmk}

 Let $(\g, [\cdot,\cdot]_\g)$ be a Lie algebra and $R$ be a modified $r$-matrix.   Semenov-Tian-Shansky  showed that $(\g, [\cdot,\cdot]_R)$ is a Lie algebra which plays important roles in the study of integrable systems \cite{STS}, where
\begin{equation}\label{defi-lie}
[x,y]_R=[R(x),y]_\g+[x,R(y)]_\g,\quad \forall x,y\in\g.
\end{equation}
 Recall that a matched pair of Lie algebras consists of Lie algebras $(\g,[\cdot,\cdot]_\g)$, $(\h,[\cdot,\cdot]_\h)$, a representation $\rho:\g\to \gl(\h)$ of $\g$ on $\h$ and a representation $\varrho:\h\to \gl(\g)$ of $\h$ on $\g$, such that some compatibility conditions are satisfied.   Bordemann further showed that the induced Lie algebra $\g_R$ represents on $\g$ which leads to a matched pair of Lie algebras $((\g,[\cdot,\cdot]_\g),(\g,[\cdot,\cdot]_R)) $ \cite{Bor}. Here we give a direct proof to be self-contained.

\begin{pro}
Let $R$ be a modified $r$-matrix on a Lie algebra $(\g, [\cdot, \cdot]_\g)$. Define a linear map $\rho:\g\lon\gl(\g)$ by
\begin{equation}\label{defi-rep}
\rho(x)y=[R(x),y]_\g-R([x,y]_\g), \quad \forall x,y\in\g.
\end{equation}
Then $\rho$ is a representation of the Lie algebra $(\g, [\cdot,\cdot]_R)$ on the vector space $\g$.
\end{pro}
\begin{proof}
For all $x, y, z\in\g$, by \eqref{assmd} and \eqref{defi-rep}, we have
\begin{eqnarray*}
&&[\rho(x), \rho(y)]z\\&=&\rho(x)\rho(y)z-\rho(y)\rho(x)z\\
&=&\rho(x)([R(y),z]_\g-R([y,z]_\g))-\rho(y)([R(x),z]_\g-R([x,z]_\g))\\
&=&[R(x),[R(y),z]_\g]_\g-[R(x),R([y,z]_\g)]_\g-R([x,[R(y),z]_\g]_\g)+R([x,R([y,z]_\g)]_\g)\\
&&-[R(y),[R(x),z]_\g]_\g+[R(y),R([x,z]_\g)]_\g+R([y,[R(x),z]_\g]_\g)-R([y,R([x,z]_\g)]_\g)\\
&=&[[R(x),R(y)]_\g,z]_\g-R([R(x),[y,z]_\g]_\g)-R([x,[R(y),z]_\g]_\g)+[x, [y,z]_\g]_\g\\
&&+R([R(y),[x,z]_\g]_\g)-[y,[x,z]_\g]_\g+R([y,[R(x),z]_\g]_\g)\\
&=&[[R(x),R(y)]_\g,z]_\g+[[x,y]_\g,z]_\g-R([[R(x),y]_\g,z]_\g)-R([[x,R(y)]_\g,z]_\g),
\end{eqnarray*}
and
\begin{eqnarray*}
&&\rho([x,y]_R)z\\&=&\rho([R(x),y]_\g+[x,R(y)]_\g)z\\
&=&[R([R(x),y]_\g+[x,R(y)]_\g),z]_\g-R([[R(x),y]_\g+[x,R(y)]_\g,z]_\g)\\
&=&[[R(x),R(y)]_\g,z]_\g+[[x,y]_\g,z]_\g-R([[R(x),y]_\g,z]_\g)-R([[x,R(y)]_\g,z]_\g).
\end{eqnarray*}
Thus we have $\rho([x,y]_R)=[\rho(x),\rho(y)]$, which means that $\rho$ is a representation of $(\g, [\cdot,\cdot]_R)$ on the vector space $\g$.
\end{proof}

Let $\dM_{\CE}^R: \Hom(\wedge^{k}\g,\g)\lon\Hom(\wedge^{k+1}\g,\g)$ be the corresponding Chevalley-Eilenberg coboundary operator of the Lie algebra $(\g, [\cdot,\cdot]_R)$ with coefficients in the representation $(\g, \rho)$. More precisely, for all $f\in\Hom(\wedge^{k}\g,\g)$ and $x_1,\cdots,x_{k+1}\in\g$, we have
\begin{eqnarray}
\label{defor-cob}&&\dM_{\CE}^{R}f(x_1,\cdots,x_{k+1})\\
\nonumber&=&\sum_{i=1}^{k+1}(-1)^{i+1}\rho(x_i)f(x_1,\cdots,\hat{x_i},\cdots,x_{k+1})\\
\nonumber&&+\sum_{i<j}(-1)^{i+j}f([x_i,x_j]_R,x_1,\cdots,\hat{x_i},\cdots,\hat{x_j},\cdots,x_{k+1})\\
\nonumber&=&\sum_{i=1}^{k+1}(-1)^{i+1}[R(x_i), f(x_1,\cdots,\hat{x_i},\cdots,x_{k+1})]_\g\\
&&-\sum_{i=1}^{k+1}(-1)^{i+1}R([x_i,f(x_1,\cdots,\hat{x_i},\cdots,x_{k+1})]_\g)\\
\nonumber&&+\sum_{i<j}(-1)^{i+j}f([R(x_i),x_j]_\g+[x_i,R(x_j)]_\g,x_1,\cdots,\hat{x_i},\cdots,\hat{x_j},\cdots,x_{k+1}).
\end{eqnarray}

Now, we define the cohomology of a modified $r$-matrix $R:\g\lon\g$. Define the space of $0$-cochains $C^{0}(R)$ to be $0$ and define the space of $1$-cochains $C^{1}(R)$ to be $\g$. For $n\geq 2$, define the space of $n$-cochains $C^{n}(R)$ by $C^{n}(R)=\Hom(\wedge^{n-1}\g,\g)$.

\begin{defi}
Let $(\g, [\cdot,\cdot]_\g)$ be a Lie algebra and $R$ be a modified r-matrix. The cohomology of  the cochain complex   $(\oplus_{i=0}^{+\infty}C^{i}(R), \dM_{\CE}^{R})$ is defined to be the {\bf cohomology for the modified $r$-matrix} $R$.
\end{defi}

Denote the set of $n$-cocycles by $Z^{n}(R)$, the set of $n$-coboundaries by $B^{n}(R)$ and the $n$-th cohomology group by
\begin{equation*}
H^{n}(R)=Z^{n}(R)/B^{n}(R), \quad n\geq 0.
\end{equation*}

It is obvious that $x\in\g$ is closed if and only if
\begin{equation*}
\ad_x\circ R=R\circ\ad_x,
\end{equation*}
and $f\in\Hom(\g,\g)$ is closed if and only if
\begin{eqnarray}\label{eqclosed}
[R(x),f(y)]_\g-R([x,f(y)]_\g)-[R(y),f(x)]_\g+R([y,f(x)]_\g)
=f([R(x),y]_\g+[x,R(y)]_\g),
\end{eqnarray}
for all $x,y\in\g$.

At the end of this section, we recall the cohomology theory of Rota-Baxter operators given in \cite{JSZ}, and establish its relation with the cohomology theory of  modified $r$-matrices.
\begin{defi}
Let $(\g, [\cdot,\cdot]_\g)$ be a Lie algebra. A linear map $B:\g\lon\g$ is called a {\bf Rota-Baxter operator of weight $\lambda$} if
\begin{equation*}
[B(x), B(y)]_\g=B([B(x),y]_\g+[x,B(y)]_\g+\lambda[x,y]_\g), \quad \forall x,y\in\g.
\end{equation*}
\end{defi}

The following result is well known.

\begin{pro}\label{corRM}
Let $\g$ be a Lie algebra and $B\in\gl(\g)$. The linear map $\Id+2B$ is a modified $r$-matrix on $\g$ if and only if $B$ is a Rota-Baxter operator of weight $1$ on $\g$.
\end{pro}


Let $B$ be a Rota-Baxter operator of weight $1$ on a Lie algebra $\g$. Consider the cochain complex $(\oplus_{k=1}^{+\infty}C^{k}(B), \dM_{\CE}^{B})$, where $C^{1}(B)=\g$ and $C^{k}(B)=\Hom(\wedge^{k-1}\g,\g)$ for $k\geq 2$, and $\dM_{\CE}^{B}$ is defined by
 \begin{eqnarray*}
&&\dM_{\CE}^Bf(u_{1},\cdots, u_{k+1})\\
&=&\sum_{i=1}^{k+1}(-1)^{i+1}B([f(u_1, \cdots, \hat{u_i}, \cdots, u_{k+1}),u_i]_{\g})\\&&+\sum_{i=1}^{k+1}(-1)^{i+1}[B(u_i), f(u_1, \cdots, \hat{u_i}, \cdots, u_{k+1})]_\g\\
&&+\sum_{i<j}(-1)^{i+j}f([B(u_i),u_j]_{\g}-[B(u_j),u_i]_{\g}+[u_i,u_j]_{\g}, u_1, \cdots, \hat{u_i}, \cdots, \hat{u_j}, \cdots, u_{k+1}),
\end{eqnarray*}
where $f\in C^{k+1}(B)$ and $u_i\in\g, 1\leq i\leq k+1.$

It was proved in \cite{JSZ} that $(\dM_{\CE}^{B})^2=0$. The cohomology of the cochain complex $(\oplus_{k=1}^{+\infty}C^{k}(B), \dM_{\CE}^{B})$ is defined to be the cohomology of the Rota-Baxter operator $B$.

 \begin{thm}
With the above notations, we have
$$
\dM_{\CE}^{R}=2\dM_{\CE}^{B}.
$$
Consequently, for $k\geq 1$, the $k$-th cohomology group $H^{k}(B)$ of a Rota-Baxter operator $B$ is isomorphic with the $k$-th cohomology group $H^{k}(R)$ of the modified $r$-matrix $R=\Id+2B$.
 \end{thm}
\begin{proof}
 For $k\geq 1,$ define linear maps $\Phi_{k}: C^{k}(B)\lon C^{k}(R)$ by $\Phi_{k}=2^{k-2}\Id$.
Then the following diagram is commutative:
\[
\xymatrix{
0 \ar[r] & \g \ar[r]^{\dM_{\CE}^{B}\quad} \ar[d]_{\frac{1}{2}\Id} & \Hom(\g,\g)  \ar[r]^{\qquad \dM_{\CE}^{B}} \ar[d]_{\Id} & \cdots \ar[r]^{\dM_{\CE}^{B}\qquad} & \Hom(\wedge^{k}\g,\g) \ar[r]^{\quad\dM_{\CE}^{B}}  \ar[d]_{2^{k-1}\Id} & \cdots\\
0 \ar[r] & \g \ar[r]^{\dM_{\CE}^{R}\quad} & \Hom(\g,\g) \ar[r]^{\qquad\dM_{\CE}^{R}} & \cdots \ar[r]^{\dM_{\CE}^{R}\qquad} & \Hom(\wedge^{k}\g,\g) \ar[r]^{\quad\dM_{\CE}^{R}} & \cdots
.}
\]
In fact, for any $f\in\Hom(\wedge^{k}\g,\g), x_i\in\g, 1\leq i\leq k+1$, we have
\begin{eqnarray*}
&&\dM_{\CE}^{R}(\Phi_{k}f)(x_1,\cdots,x_{k+1})\\
&=&2^{k-1}\Big(\sum_{i=1}^{k+1}(-1)^{i+1}([x_i, f(x_1,\cdots,\hat{x_i},\cdots,x_{k+1})]_\g\\&&+2[B(x_i),f(x_1,\cdots,\hat{x_i},\cdots,x_{k+1})]_{\g})\\
&&-\sum_{i=1}^{k+1}(-1)^{i+1}[x_i,f(x_1,\cdots,\hat{x_i},\cdots,x_{k+1})]_\g\\&&-\sum_{i=1}^{k+1}(-1)^{i+1}2B([x_i,f(x_1,\cdots,\hat{x_i},\cdots,x_{k+1})]_\g)\\
&&+\sum_{i<j}(-1)^{i+j}2f([x_i,x_j]_{\g},x_1,\cdots,\hat{x_i},\cdots,\hat{x_j},\cdots,x_{k+1})\\
&&+\sum_{i<j}(-1)^{i+j}2f([B(x_i),x_j]_{\g}+[x_i,B(x_j)]_{\g},x_1,\cdots,\hat{x_i},\cdots,\hat{x_j},\cdots,x_{k+1})\Big)\\
&=&2^{k}\Big(\sum_{i=1}^{k+1}(-1)^{i+1}([B(x_i),f(x_1,\cdots,\hat{x_i},\cdots,x_{k+1})]_{\g}-B([x_i,f(x_1,\cdots,\hat{x_i},\cdots,x_{k+1})]_\g))\\
&&+\sum_{i<j}(-1)^{i+j}f([B(x_i),x_j]_{\g}+[x_i,B(x_j)]_{\g}+[x_i,x_j]_{\g},x_1,\cdots,\hat{x_i},\cdots,\hat{x_j},\cdots,x_{k+1})\Big)\\
&=&\Phi_{k+1}(\dM_{\CE}^{B}f)(x_1,\cdots, x_{k+1}),
\end{eqnarray*}
which implies that $\dM_{\CE}^{R}=2\dM_{\CE}^{B} $ and $H^{k}(B)\cong H^{k}(R), k\geq 1$.
\end{proof}

\begin{ex}\label{H1}{\rm
Consider the Lie algebra $\g=\mathfrak{sl}(n, \mathbb{R})$. It is well known that the Cartan subalgebra of  $\mathfrak{sl}(n, \mathbb{R})$ is $H=\text{span}\{E_{ii}-E_{i+1i+1}|1\leq i\leq n-1\}$. Denote the Borel subalgebra of $\mathfrak{sl}(n, \mathbb{R})$ by $B(\mathfrak{sl}(n, \mathbb{R}))$. It is well known that $B(\mathfrak{sl}(n, \mathbb{R}))=H\oplus\text{span}\{E_{ij}|i<j\}$. Thus $\mathfrak{sl}(n, \mathbb{R})=B(\mathfrak{sl}(n, \mathbb{R}))\oplus A$ as vector spaces, where $A=\text{span}\{E_{ij}|i>j\}$. Define a linear map $R: \mathfrak{sl}(n, \mathbb{R})\lon\mathfrak{sl}(n, \mathbb{R})$ by
\begin{equation*}
R(x+u)=x-u, \quad \forall x\in B(\mathfrak{sl}(n, \mathbb{R})), u\in A.
\end{equation*}
By Remark \ref{thmex1}, we obtain that $R$ is a modified $r$-matrix on the Lie algebra $\mathfrak{sl}(n, \mathbb{R})$. Assume that $a=x+u\in \mathfrak{sl}(n, \mathbb{R})$ where $x\in B(\mathfrak{sl}(n, \mathbb{R}))$ and $u\in A$, such that $\dM_{\CE}^{R}a=0$, that is
\begin{equation*}
\dM_{\CE}^{R}a(y)=[R(y), a]-R([y, a])=0, \quad \forall y\in\mathfrak{sl}(n, \mathbb{R}).
\end{equation*}
\begin{itemize}
  \item For any $y\in A$, $[R(y), a]-R([y, a])=0$ implies that $x\in H$.

  \item For any $y\in B(\mathfrak{sl}(n, \mathbb{R}))$,   $[R(y), a]-R([y, a])=0$ implies that $u=0$.
\end{itemize}
Thus   $\dM_{\CE}^{R}a=0$ if and only if $a\in H$. Therefore,   $H^{1}(R)\cong \mathbb{R}^{n-1}$.
}
\end{ex}

\begin{ex}{\rm
Consider the Lie algebra $\g=\mathfrak{sl}(2, \mathbb{R})$, where the Lie bracket is given by $[e, f]=h, [h, e]=2e$ and $[h, f]=-2f$ with respect to the basis $\{e, f, h\}$. Then $R:\mathfrak{sl}(2, \mathbb{R})\lon\mathfrak{sl}(2, \mathbb{R})$ defined by
\begin{equation*}
R(e, f, h)=(e, f, h)\left(
                      \begin{array}{ccc}
                        1 & 0 & 0 \\
                        0 & -1 & 0 \\
                        0 & 0 & 1 \\
                      \end{array}
                    \right),
\end{equation*}
is a modified $r$-matrix.
Let $T=\left(
                      \begin{array}{ccc}
                        t_{11} & t_{12} & t_{13} \\
                        t_{21} & t_{22} & t_{23} \\
                        t_{31} & t_{32} & t_{33} \\
                      \end{array}
                    \right): \mathfrak{sl}(2, \mathbb{R})\lon\mathfrak{sl}(2, \mathbb{R})$ satisfy
 $\dM_{\CE}^{R}T=0$. Then we obtain
\begin{eqnarray*}
~0&=&[e, T(f)]+[f, T(e)]-R([e, T(f)])+R([f, T(e)]),\\
~0&=&[e, T(h)]-[h, T(e)]-R([e, T(h)])+R([h, T(e)])+4T(e)
\end{eqnarray*}
and
\begin{equation*}
0=-[f, T(h)]-[h, T(f)]-R([f, T(h)])+R([h, T(f)]).
\end{equation*}
Thus we have $t_{11}=t_{21}=t_{31}=0$ and $t_{22}=t_{13}=0$. By Example \ref{H1}, we have $B^{2}(R)=\mathrm{Im}\dM_{\CE}^{R}\cong\frac{\g}{\ker\dM_{\CE}^{R}}=\frac{\g}{H^{1}(R)}\cong \mathbb{R}^{2}$. Thus $H^{2}(R)\simeq \mathbb{R}^{2}$.
}
\end{ex}

\section{Algebraic deformations of modified $r$-matrices}\label{sec:alg}

In this section, we construct a differential graded Lie algebra that governs algebraic deformations of a modified $r$-matrix.

Let $(\g,[\cdot,\cdot]_\g)$ be a Lie algebra. We consider the graded vector space $C^*(\g)=\oplus_{k=1}^{+\infty}\Hom(\wedge^{k}\g,\g)$. Define a skew-symmetric bracket operation
 \begin{equation*}
 \Courant{\cdot,\cdot}: \Hom(\wedge^{p}\g,\g)\times\Hom(\wedge^{q}\g,\g)\lon\Hom(\wedge^{p+q}\g,\g)
 \end{equation*}
 by
\begin{eqnarray}\label{defiCour}
&&\Courant{f,g}(x_1,x_2,\cdots,x_{p+q})\\
\nonumber&=&\sum_{\sigma\in S(q,1,p-1)}(-1)^{\sigma}f([g(x_{\sigma(1)},\cdots,x_{\sigma(q)}),x_{\sigma(q+1)}]_{\g},x_{\sigma(q+2)},\cdots,x_{\sigma(p+q)})\\
\nonumber&&-(-1)^{pq}\sum_{\sigma\in S(p,1,q-1)}(-1)^{\sigma}g([f(x_{\sigma(1)},\cdots,x_{\sigma(p)}),x_{\sigma(p+1)}]_{\g}, x_{\sigma(p+2)},\cdots,x_{\sigma(p+q)})\\
\nonumber&&+(-1)^{pq}\sum_{\sigma\in S(p,q)}(-1)^{\sigma}[f(x_{\sigma(1)},\cdots,x_{\sigma(p)}),g(x_{\sigma(p+1)},\cdots,x_{\sigma(p+q)})]_{\g},
\end{eqnarray}
for all $f\in\Hom(\wedge^{p}\g,\g), g\in\Hom(\wedge^{q}\g,\g)$.
\emptycomment{
In fact, the bracket $\Courant{\cdot,\cdot}$ can be given intrinsically using the Nijenhuis-Richardson bracket $[\cdot,\cdot]_{\NR}$ \cite{NR,NR2} via the derived bracket \cite{Kosmann-Schwarzbach} by
\begin{equation}\label{defiCour1}
 \Courant{f,g}=(-1)^{p}[[\pi,f]_\NR,g]_\NR,
\end{equation}
where we denote the Lie bracket $[\cdot,\cdot]_{\g}$ by $\pi$.
}

Then we have the following theorem characterizing modified $r$-matrices.

\begin{thm}\label{cormdfi}
Let $(\g, [\cdot,\cdot]_{\g})$ be a Lie algebra. Then $(C^*(\g),\Courant{\cdot,\cdot})$ is a graded Lie algebra and its Maurer-Cartan elements are precisely  Rota-Baxter operators of weight $0$.

Moreover, a linear map $R\in\gl(\g)$ is a modified $r$-matrix on the Lie algebra $\g$ if and only if $R$ satisfies the equation
\begin{equation}\label{defieqmdfi}
\Courant{R,R}=2\pi.
\end{equation}
where denote $[\cdot, \cdot]_\g$ by $\pi$.
\end{thm}
\begin{proof}
By \cite[Corollary 6.1]{TBGS}, $(C^*(\g),\Courant{\cdot,\cdot})$ is a graded Lie algebra.

 For $R\in\gl(\g)$, we have
\begin{equation*}
\Courant{R,R}(x,y)=2(R([R(x),y]_{\g})-R([R(y),x]_{\g})-[R(x),R(y)]_{\g}),\quad \forall x,y\in\g.
\end{equation*}
By this equality, we can deduce that on the one hand $R$ is a Rota-Baxter operator of weight $0$ if and only if $\Courant{R,R}=0$, i.e. $R$ is a Maurer-Cartan element. On the other hand, $R$ is a modified $r$-matrix on the Lie algebra $\g$ if and only if $R$ satisfies   \eqref{defieqmdfi}.
\end{proof}

\begin{pro}\label{propd}
Let $R$ be a modified $r$-matrix on a Lie algebra $(\g, [\cdot,\cdot]_{\g})$. Then $\Courant{R,R}$ is in the center of the graded Lie algebra $(C^{*}(\g), \Courant{\cdot,\cdot})$.
\end{pro}
\begin{proof}
\emptycomment{
 For all $f\in\Hom(\wedge^{k}\g,\g)$, by   \eqref{defiCour1} and Theorem \ref{cormdfi}, we have
\begin{eqnarray*}
\Courant{\Courant{R,R},f}&=&\Courant{2\pi,f}=2[[\pi,\pi]_\NR,f]_\NR=0,
\end{eqnarray*}
which implies that $\Courant{R,R}$ is in the center of $C^{*}(\g)$.
}
Denote the Lie bracket $[\cdot,\cdot]_{\g}$ by $\pi$. Since $R$ is a modified $r$-matrix on the Lie algebra $\g$, we have $\Courant{R,R}=2\pi$ via Theorem \ref{cormdfi}. For all $f\in\Hom(\wedge^{k}\g,\g)$, by \eqref{defiCour}, we have
\begin{eqnarray*}
&&\Courant{2\pi,f}(x_1,\cdots,x_k,x_{k+1},x_{k+2})\\
&=&2\Big(\sum_{\sigma\in S(k,1,1)}(-1)^{|\sigma|}\pi(\pi(f(x_{\sigma(1)},\cdots,x_{\sigma(k)}),x_{\sigma(k+1)}),x_{\sigma(k+2)})\\
&&-\sum_{\sigma\in S(2,1,k-1)}(-1)^{|\sigma|}f(\pi(\pi(x_{\sigma(1)},x_{\sigma(2)}),x_{\sigma(3)}),x_{\sigma(4)},\cdots,x_{\sigma(k+2)})\\
&&+\sum_{\sigma\in S(2,k)}(-1)^{|\sigma|}\pi(\pi(x_{\sigma(1)},x_{\sigma(2)}),f(x_{\sigma(3)},\cdots,x_{\sigma(k+2)}))\Big)\\
&=&0,
\end{eqnarray*}
which implies that $\Courant{R,R}$ is in the center of $C^{*}(\g)$.
\end{proof}

We denote $\Courant{R,\cdot}$ by $\dM_{R}$. Now we obtain the differential graded Lie algebra that governs algebraic deformations of a modified $r$-matrix.

\begin{thm}
With the above notations, $(C^{*}(\g), \Courant{\cdot,\cdot}, \dM_{R})$ is a differential graded Lie algebra.

Furthermore, $R+R'$ is still a modified $r$-matrix on the Lie algebra $(\g, [\cdot,\cdot]_\g)$ if and only if $R'$ is a Maurer-Cartan element of the differential graded Lie algebra $(C^{*}(\g), \Courant{\cdot,\cdot}, \dM_{R})$.
\end{thm}

\begin{proof}
  It follows from the graded Jacobi identity that $\dM_{R}$ is a graded derivation on the graded Lie algebra $(C^{*}(\g), \Courant{\cdot,\cdot})$. By Proposition \ref{propd}, we have
$$\dM_{R}^{2}f=\Courant{R,\Courant{R,f}}=\Courant{\Courant{R,R},f}-\Courant{R,\Courant{R,f}},$$
which implies that
\begin{equation*}
\dM_{R}^{2}f=\Courant{R,\Courant{R,f}}=\frac{1}{2}\Courant{\Courant{R,R},f}=0.
\end{equation*}
 Therefore, $(C^{*}(\g), \Courant{\cdot,\cdot}, \dM_{R})$ is a differential graded Lie algebra.

Let $R'$ be a linear map from $\g$ to $\g$. Then $R+R'$ is a modified $r$-matrix if and only if
\begin{equation*}
\Courant{R+R',R+R'}=2\pi,
\end{equation*}
that is
\begin{equation*}
0=\Courant{R,R'}+\frac{1}{2}\Courant{R',R'}=\dM_{R}R'+\frac{1}{2}\Courant{R',R'}.
\end{equation*}
Thus $R+R'$ is still a modified $r$-matrix on the Lie algebra $(\g, [\cdot,\cdot]_\g)$ if and only if $R'$ is a Maurer-Cartan element of the differential graded Lie algebra $(C^{*}(\g), \Courant{\cdot,\cdot}, \dM_{R})$.
\end{proof}

At the end of this section, we establish the relationship between the coboundary operator $\dM_{\CE}^{R}$ and the differential $\dM_R$.

\begin{pro}
Let $R$ be a modified $r$-matrix on a Lie algebra $(\g, [\cdot,\cdot]_{\g})$. Then we have
\begin{equation*}
\dM_{\CE}^{R}(f)=(-1)^{n-1}\Courant{R,f}, \quad \forall f\in\Hom(\wedge^{n-1}\g,\g).
\end{equation*}
\end{pro}
\begin{proof}
For any $f\in\Hom(\wedge^{n-1}\g,\g)$ and $x_i,1\leq i\leq n$, by \eqref{defiCour}, we have
\begin{eqnarray*}
&&(-1)^{n-1}\Courant{R,f}(x_1,\cdots,x_n)\\
&=&(-1)^{n-1}\Big(\sum_{\sigma\in S(n-1,1)}(-1)^{\sigma}R([f(x_{\sigma(1)},\cdots,x_{\sigma(n-1)}),x_{\sigma(n)}]_{\g})\\
&&-(-1)^{n-1}\sum_{\sigma\in S(1,1,n-2)}(-1)^{\sigma}f([R(x_{\sigma(1)}),x_{\sigma(2)}]_{\g},x_{\sigma(3)},\cdots,x_{\sigma(n)})\\
&&+(-1)^{n-1}\sum_{\sigma\in S(1,n-1)}(-1)^{\sigma}[R(x_{\sigma(1)}),f(x_{\sigma(2)},\cdots,x_{\sigma(n)})]_{\g}\Big)\\
&=&\sum_{i=1}^{n}(-1)^{i+1}R([f(x_{1},\cdots,\hat{x_i},\cdots,x_{n}),x_{i}]_{\g})\\
&&+\sum_{i<j}(-1)^{i+j}\Big(f([R(x_{i}),x_{j}]_{\g},x_{1},\cdots,\hat{x_i},\cdots,\hat{x_j},\cdots,x_{n})\\
&&-f([R(x_{j}),x_{i}]_{\g},x_{1},\cdots,\hat{x_i},\cdots,\hat{x_j},\cdots,x_{n})\Big)\\
&&+\sum_{i=1}^{n}(-1)^{i+1}[R(x_i),f(x_{1},\cdots,\hat{x_{i}},\cdots,x_{n})]_{\g}\\
&=&\dM_{\CE}^{R}(f)(x_1,\cdots,x_n).
\end{eqnarray*}
We finish the proof.
\end{proof}

\section{Geometric deformations of modified $r$-matrices}\label{sec:geo}

In this section, we study geometric deformations of modified $r$-matrices following the approach developed by Crainic,   Schatz and  Struchiner. We show that the condition $H^{2}(R)=0$ will imply certain rigidity theorem, and the condition $H^{3}(R)=0$ will imply the space of modified $r$-matrices on the Lie algebra $\g$ is a manifold in a neighborhood of $R$. We also give the necessary and sufficient condition on a 2-cocycle giving a geometric deformation using the Kuranishi map.

\begin{defi}
Let $R$ be a modified $r$-matrix on a Lie algebra $(\g, [\cdot, \cdot]_\g)$. A {\bf geometric deformation} of $R$ is a smooth one parameter family of modified $r$-matrices $R_t$ on the Lie algebra $(\g, [\cdot, \cdot]_\g)$ such that $R_0=R$.
\end{defi}

\begin{defi}
Two geometric deformations $R_t$ and $R_t'$ of $R$ are called {\bf equivalent} if there exists a smooth family of modified $r$-matrices isomorphism $\varphi_t: R_t\lon R_t'$ such that $\varphi_0=\Id$, where $\varphi_t$ are inner automorphisms of the Lie algebra $\g$.
\end{defi}

Let $R_t$ be a geometric deformation of $R$. Denote $\frac{d}{dt}|_{t=0}R_t$ by $\dot{R_0}$. Then there is the following proposition.
\begin{pro}\label{protangentgamma}
With the above notations, $\dot{R_0}$ is a $2$-cocycle in $C^{2}(R)$. Moreover if $R_t$ and $R_t'$ are equivalent geometric deformations of $R$, then $[\dot{R_0}]=[\dot{R_0'}]$ in $H^{2}(R)$.
\end{pro}
\begin{proof}
Since $R_t$ is a geometric deformation of $R$, for any $x, y\in\g$, we have
\begin{eqnarray}
\label{eqR'}&&[R(x), \dot{R_0}(y)]_\g+[\dot{R_0}(x), R(y)]_\g\\
\nonumber&=&\frac{d}{dt}|_{t=0}[R_t(x), R_t(y)]_\g\\
\nonumber&=&\frac{d}{dt}|_{t=0}(R_t([R_t(x), y]_\g+[x, R_t(y)]_\g)-[x, y]_\g)\\
\nonumber&=&\dot{R_0}([R(x), y]_\g+[x, R(y)]_\g)+R([\dot{R_0}(x), y]_\g+[x, \dot{R_0}(y)]_\g).
\end{eqnarray}
Thus by \eqref{eqclosed} and \eqref{eqR'}, we have $\dM^{R}_{\CE}(\dot{R_0})=0$.

Assume that $\varphi_t$ is an isomorphism from $R_t$ to $R'_t$,  that is
\begin{equation*}
\varphi_t(R_t(x))=R'_t(\varphi_t(x)), \quad \forall x\in\g.
\end{equation*}
Denote $\frac{d}{dt}|_{t=0}\varphi_t$ by $\dot{\varphi_0}$. Then we have $\dot{\varphi_0}(R(x))+\dot{R_0}(x)=R(\dot{\varphi_0}(x))+\dot{R'_0}(x)$.
Since $\varphi_t$ are inner automorphisms of the Lie algebra $\g$, it follows that $\dot{\varphi_0}$ is an inner derivation of the Lie algebra $\g$. Thus there exists $y\in\g$ such that $\dot{\varphi_0}=\ad_y$. Therefore, we have
\begin{equation*}
[y, R(x)]_\g+\dot{R_0}(x)=R([y, x]_\g)+\dot{R'_0}(x),
\end{equation*}
which implies $\dot{R_0}-\dot{R'_0}=\dM^{R}_{\CE}(y)$. Thus $[\dot{R_0}]=[\dot{R_0'}]$ in $H^{2}(R)$.
\end{proof}

Next, we consider under which conditions does a cocycle $f\in Z^{2}(R)$ determine a geometric deformation $R_t$. Define the Kuranishi map $K: Z^{2}(R)\lon H^{3}(R)$ by
\begin{equation*}
K(f)=[\Courant{f,f}], \quad \forall f\in Z^{2}(R).
\end{equation*}

Now we give a necessary condition of the above question. The sufficient condition need some preparations and  will be given at the end of this section.
\begin{pro}\label{pro:K}
Assume that there exists a geometric deformation $R_t$ of $R$ on a Lie algebra $(\g, [\cdot, \cdot]_\g)$ such that $\dot{R_0}=f\in Z^{2}(R)$, then $K(f)=0$.
\end{pro}
\begin{proof}
Consider the Taylor expansion of $R_t$ around $t=0$, then we have
\begin{equation*}
R_t(x)=R(x)+tf(x)+\frac{t^{2}}{2}g(x)+o(t^{3}).
\end{equation*}
Since $[R_t(x), R_t(y)]_\g=R_t([R_t(x), y]_\g+[x, R_t(y)]_\g)-[x, y]_\g$ and $f\in Z^{2}(R)$, we have
\begin{equation}\label{Taylor}
\frac{t^{2}}{2}\Big(\dM^{R}_{\CE}(g)(x, y)+2[f(x), f(y)]_\g-2f([f(x), y]_\g+[x, f(y)]_\g)\Big)+o(t^3)=0.
\end{equation}
Thus by \eqref{defiCour} and \eqref{Taylor}, we obtain $\Courant{f, f}=\dM^{R}_{\CE}g$, which implies that $K(f)=0$.
\end{proof}

Let $E\stackrel{\pi}{\lon} M$ be a vector bundle. Assume that there is a smooth action $\cdot: G\times E\lon E$ of a Lie group $G$ on $E$ preserving the zero-section $Z: M\lon E$. It follows that $M$ inherits a $G$-action. We also denote the action of $G$ on $M$ by $\cdot: G\times M\lon M$. For all $x\in M$, define a smooth map $\mu_x: G\lon M$ by $\mu_x(g)=g\cdot x$. Denote the tangent map from $\g$ to $T_xM$ by $D(\mu_x)_{e_G}$, where $e_G$ is the unit of $G$.
\begin{defi}\rm(\cite{CSS})
A section $s: M\lon E$ is called {\bf equivariant} if $s$ satisfies
\begin{equation*}
s(g\cdot x)=g\cdot s(x), \quad \forall g\in G, x\in M.
\end{equation*}
\end{defi}

Denote the zero set of a section $s: M\lon E$ by $z(s)=\{x\in M| s(x)=0\}$. A zero $x\in M$ of $s$ is called {\bf non-degenerate} if the sequence
\begin{equation*}
\g\stackrel{D(\mu_x)_{e_G}}{\longrightarrow}T_xM\stackrel{D^{v}(s)_x}{\longrightarrow} E_x
\end{equation*}
is exact, where $D^{v}(s)_x$ is the vertical derivative of $s$ at $x$.

\begin{pro}\rm(\cite{CSS})\label{proorbR}
Let $s$ be an equivariant section of the vector bundle $E\stackrel{\pi}{\lon} M$ and $x$ be a non-degenerate zero of $s$. Then there is an open neighborhood $U$ of $x$ and a smooth map $p: U\lon G$ such that for all $m\in U$ with $s(m)=0$, one has $p(m)\cdot x=m$. In particular, the orbit of $x$ under the action of $G$ and the zero set of $s$ coincide in an open neighborhood of $x$.
\end{pro}

\begin{pro}\rm(\cite{CSS})\label{proorbR1}
Let $E$ and $F$ be vector bundles over a smooth manifold $M$. Let $s\in \Gamma(E)$ be a section and $\phi\in\Gamma(\Hom(E, F))$ be a vector bundle map such that $\phi\circ s=0$. Suppose that $x\in M$ is $s(x)=0$ such that
\begin{equation*}
T_xM\stackrel{D^{v}(s)_x}{\longrightarrow} E_x\stackrel{\phi_x}{\longrightarrow} F_x
\end{equation*}
is exact. Then $s^{-1}(0)$ is locally a manifold around $x$ of dimension $\mathrm{dim}\ker(D^{v}(s)_x)$.
\end{pro}

Denote the group whose elements are inner automorphisms of a Lie algebra $\g$ by $\InnAut(\g)$. Then its Lie algebra is the Lie algebra of inner derivations of $\g$ and denote it by $\InnDer(\g)$. Define an action of $\InnAut(\g)$ on $\Hom(\g, \g)$ by
\begin{equation*}
\cdot: \InnAut(\g)\times\Hom(\g, \g)\lon\Hom(\g, \g), \quad A\cdot f=AfA^{-1},
\end{equation*}for all $A\in\InnAut(\g), f\in\Hom(\g, \g).$
Assume that $R$ is a modified $r$-matrix on a Lie algebra $\g$, then the orbit $\Orb_R=\{A\cdot R| A\in\InnAut(\g) \}$ of $R$ is a manifold. Define a map $\mu_R: \InnAut(\g)\lon \Hom(\g, \g)$ by $\mu_R(A)=A\cdot R$. Then $T_R\Orb_R$ is $D(\mu_R)_{e_G}(\InnDer(\g))$, where $D(\mu_R)_{e_G}$ is the tangent map of $\mu_R$ at $e_G$.

\begin{pro}\label{protang}
With the above notations, $T_R\Orb_R$ is $B^{2}(R)$.
\end{pro}
\begin{proof}
Since $T_R\Orb_R=D(\mu_R)_{e_G}(\InnDer(\g))$, for any $v\in T_R\Orb_R$, there exists $x\in\g$ such that
\begin{eqnarray*}
v&=&\frac{d}{dt}|_{t=0}\exp(t\ad_x)\cdot R\\
&=&\frac{d}{dt}|_{t=0}(\exp(t\ad_x)R\exp(-t\ad_x))\\
&=&\ad_xR-R\ad_x\\
&=&-\dM^{R}_{\CE}x.
\end{eqnarray*}
Thus we have $T_R\Orb_R=B^{2}(R)$.
\end{proof}

\begin{thm}
Let $R$ be a modified $r$-matrix on a Lie algebra $(\g, [\cdot, \cdot]_\g)$. If $H^{2}(R)=0$, then there exists an open neighborhood $U\subset \Hom(\g, \g)$ of $R$ and a smooth map $p: U\lon \InnAut(\g)$ such that $p(R')\cdot R=R'$ for every modified $r$-matrix $R'\in U$.
\end{thm}
\begin{proof}
Denote by $M=\Hom(\g, \g)$ and $E=\Hom(\g, \g)\times\Hom(\wedge^{2}\g, \g)$. Then $E$ is a trivial vector bundle over $M$ with fiber $\Hom(\wedge^{2}\g, \g)$. Define an action of $\InnAut(\g)$ on the manifold $E$ by
\begin{equation*}
\cdot: \InnAut(\g)\times E\lon E, \quad A\cdot(f, \alpha)=(AfA^{-1}, A\alpha\circ A^{-1}),
\end{equation*}
for $(f, \alpha)\in E,~ A\in\InnAut(\g),$ where $A\alpha\circ A^{-1}(x, y)=A\alpha(A^{-1}x, A^{-1}y)$ for any $x, y\in\g$. Define a section $s: M\lon E$ by
\begin{equation*}
s(f)=(f, S(f)), \quad \forall f\in M,
\end{equation*}
where $S: \Hom(\g, \g)\lon\Hom(\wedge^{2}\g, \g)$ is given by
\begin{equation*}
S(f)(x, y)=[f(x), f(y)]_\g-f([f(x), y]_\g+[x, f(y)]_\g)+[x, y]_\g,
\end{equation*}for all $f\in\Hom(\g, \g), ~x, y\in\g.$
Then for any $A\in \InnAut(\g), f\in M$ and $x, y\in\g$, we have
\begin{eqnarray*}
&&AS(f)\circ A^{-1}(x, y)\\
&=&A([f(A^{-1}x), f(A^{-1}y)]_\g-f([f(A^{-1}x), A^{-1}y]_\g+[A^{-1}x, f(A^{-1}y)]_\g)+[A^{-1}x, A^{-1}y]_\g)\\
&=&[Af(A^{-1}x), Af(A^{-1}y)]_\g-AfA^{-1}([Af(A^{-1}x), y]_\g+[x, Af(A^{-1}y)]_\g)+[x, y]_\g\\
&=&S(AfA^{-1})(x, y).
\end{eqnarray*}
Thus we have
$$
A\cdot s(f)=(AfA^{-1}, AS(f)\circ A^{-1})=s(A\cdot f),
$$
which implies that $s$ is an equivariant section.

Since $R$ is a modified $r$-matrix on the Lie algebra $\g$, it follows that $R\in z(s)$. Moreover, since $E$ is a trivial vector bundle, we have $D^{v}(s)_R=D(S)_R: T_RM\lon E_R$. For any $g\in\Hom(\g, \g), x, y\in\g$,
\begin{eqnarray}\label{defiDS}
&&D(S)_R(g)(x, y)\\
\nonumber&=&\frac{d}{dt}|_{t=0}S(R+tg)(x, y)\\
\nonumber&=&\frac{d}{dt}|_{t=0}\Big([R(x)+tg(x), R(y)+tg(y)]_\g\\\nonumber&&-(R+tg)([R(x)+tg(x), y]_\g+[x, R(y)+tg(y)]_\g)+[x, y]_\g\Big)\\
\nonumber&=&[g(x), R(y)]_\g+[R(x), g(y)]_\g-R([g(x), y]_\g+[x, g(y)]_\g)-g([R(x), y]_\g+[x, R(y)]_\g)\\
\nonumber&=&\dM_{\CE}^{R}(g)(x, y).
\end{eqnarray}

By Proposition \ref{protang} and $H^{2}(R)=0$, we have that $R$ is a non-degenerate zero of $s$. By Proposition \ref{proorbR}, there exists an open neighborhood $U\subset \Hom(\g, \g)$ of $R$ and a smooth map $p: U\lon \InnAut(\g)$ such that $p(R)\cdot R=R'$ for every modified $r$-matrix $R'\in U$.
\end{proof}

\begin{thm}
Let $R$ be a modified $r$-matrix on a Lie algebra $(\g, [\cdot, \cdot]_\g)$. If $H^{3}(R)=0$, then the space of modified matrices on the Lie algebra $\g$ is a manifold in a neighborhood of $R$, whose  dimension is $\mathrm{dim}Z^{2}(R)$.
\end{thm}
\begin{proof}
Denote by $M=\Hom(\g, \g), E=\Hom(\g, \g)\times\Hom(\wedge^{2}\g, \g)$ and $F=\Hom(\g, \g)\times\Hom(\wedge^{3}\g, \g)$. Then $E$ and $F$ are trivial vector bundles over $M$ with fiber $\Hom(\wedge^{2}\g, \g)$ and $\Hom(\wedge^{3}\g, \g)$ respectively. Define a smooth map $\phi: E\lon F$ by
\begin{equation*}
\phi(f, \alpha)=(f, \Courant{f, \alpha}), \quad \forall f\in M, \alpha\in\Hom(\wedge^{2}\g, \g).
\end{equation*}
Thus $\phi$ is a vector bundle map.

Moreover, denote the Lie bracket $[\cdot, \cdot]_\g$ by $\pi$, define $s(f)=\pi-\frac{1}{2}\Courant{f, f}$. By Proposition \ref{propd}, we know that $\pi$ lies in the center, we have $\phi\circ s(f)=(f, \Courant{f, \pi}-\frac{1}{2}\Courant{f, \Courant{f, f}})=(f, 0)$, which implies $\phi\circ s=0$. Moreover, denote $\phi_R: E_R\lon F_R$ by $\phi(R, \cdot)$, then $\phi_R=\dM_{\CE}^{R}$. By \eqref{defiDS} and $H^{3}(R)=0$, we have that
$$
T_RM\stackrel{D^{v}(s)_R}{\longrightarrow} E_R\stackrel{\phi_R}{\longrightarrow} F_R
$$
is exact. By Proposition \ref{proorbR1}, we obtain that the space of modified $r$-matrices on the Lie algebra $\g$ is a manifold in a neighborhood of $R$, whose dimension is $\mathrm{dim}Z^{2}(R)$.
\end{proof}

At the end of this section, we give the sufficient condition on a 2-cocycle to give a geometric deformation. Recall that the necessary condition is given in Proposition \ref{pro:K} using the Kuranishi map.

\begin{cor}
With the above notations, if $H^{3}(R)=0$, then any $f\in Z^{2}(R)$ gives rise to a geometric deformation of $R$.
\end{cor}
\begin{proof}
Since $R$ is a modified $r$-matrix and $H^{3}(R)=0$, we have that the space $W$ of modified $r$-matrices on the Lie algebra $\g$ is a manifold in a neighborhood of $R$, whose dimension is $\mathrm{dim}Z^{2}(R)$. Assume $\gamma(t)\in W$, by Proposition \ref{protangentgamma}, we have $\dot{\gamma}(0)\in Z^{2}(R)$. Moreover, $\mathrm{dim}W=\mathrm{dim}Z^{2}(R)$, then $T_RW=Z^{2}(R)$. Thus any $f\in Z^{2}(R)$ gives rise to a geometric deformation of $R$.
\end{proof}

\section{Linear deformations of modified $r$-matrices}\label{sec:lin}

In this section, we study linear deformations of a modified $r$-matrix using the established cohomology theory. In particular, a trivial linear deformation leads to a Nijenhuis element for a modified $r$-matrix $R$.

\begin{defi}\label{cohoR}
Let $R$ be a modified $r$-matrix on the Lie algebra $(\g, [\cdot,\cdot]_\g)$ and $\hat{R}:\g\lon\g$ be a linear map. If there exists a positive number $\epsilon\in\mathbb{R}$ such that $R_t=R+t\hat{R}$ is still a modified $r$-matrix on the Lie algebra $(\g, [\cdot,\cdot]_\g)$ for all $t\in (-\epsilon, \epsilon)$, we say that $\hat{R}$ generates a {\bf linear deformation} of the modified $r$-matrix $R$.
\end{defi}

\begin{defi}
Let $R:\g\lon\g$ be a modified r-matrix on $\g$. Two linear deformations $R_t^{1}=R+t\hat{R}_1$ and $R_t^{2}=R+t\hat{R}_2$ are said to be {\bf equivalent} if there exists an $x\in\g$ such that
$$\varphi_t=\Id_\g+t\ad_x,$$
satisfies the following conditions:
\begin{itemize}
\item[\rm(i)] $\varphi_{t}([y, z]_\g)=[\varphi_{t}(y),\varphi_{t}(z)]_\g,\quad \forall y,z\in\g,$
\item[\rm(ii)] $R^{2}_{t}\circ\varphi_{t}=\varphi_{t}\circ R^{1}_{t}$.
\end{itemize}
\end{defi}

\begin{thm}
Let $\hat{R}:\g\lon\g$ generate a   linear deformation of the modified $r$-matrix $R$. Then  $\hat{R}$ is a $2$-cocycle.

 Let $R^{1}_{t}$ and $R^{2}_{t}$ be equivalent linear deformations of $R$ generated by $\hat{R_1}$ and $\hat{R_2}$ respectively. Then $[\hat{R_1}]=[\hat{R_2}]$ in $H^{2}(R)$.
\end{thm}
\begin{proof}

Since $R_t=R+t\hat{R}$ is a modified $r$-matrix on the Lie algebra $(\g, [\cdot,\cdot]_\g)$, we have
\begin{equation*}
[R_t(x),R_t(y)]_\g=R_t([R_t(x), y]_\g+[x, R_t(y)]_\g)-[x, y]_{\g},\quad \forall x,y\in\g.
\end{equation*}
Consider the coefficients of $t$ and $t^2$ respectively, we have
\begin{eqnarray}\label{defi-defor2}
&&[\hat{R}(x),R(y)]_\g+[R(x),\hat{R}(y)]_\g\\
\nonumber&=&R([\hat{R}(x),y]_\g+[x,\hat{R}(y)]_\g)+\hat{R}([R(x),y]_\g+[x,R(y)]_\g),\quad \forall x,y\in\g,
\end{eqnarray}
and
\begin{eqnarray}\label{defi-Defor}
[\hat{R}(x), \hat{R}(y)]_\g=\hat{R}([\hat{R}(x), y]_\g+[x, \hat{R}(y)]_\g).
\end{eqnarray}
 By \eqref{defi-defor2}, we deduce  that $\hat{R}$ is a 2-cocycle  of the modified $r$-matrix $R$.

If $R^{1}_{t}$ and $R^{2}_{t}$ are equivalent linear deformations of $R$, then there exists $x\in \g$ such that
\begin{equation*}
(\Id_\g+t\ad_x)(R+t\hat{R}_1)(u)=(R+t\hat{R}_2)(\Id_\g+t\ad_x)(u), \quad \forall u\in\g,
\end{equation*}
which implies
\begin{equation}
\label{equa-defor}\hat{R}_1(u)-\hat{R}_2(u)=[R(u),x]_\g-R([u,x]_\g),\quad \forall u\in\g.
\end{equation}

By \eqref{equa-defor}, we have
\begin{equation*}
\hat{R_{1}}-\hat{R_{2}}=\dM_{\CE}^{R}x,
\end{equation*}
where $\dM_{\CE}^{R}$ is given by \eqref{defor-cob}. Thus $[\hat{R_1}]=[\hat{R_2}]$ in $H^{2}(R)$.
\end{proof}

\begin{defi}
 A linear deformation of a modified $r$-matrix $R$ generated by   $\hat{R}$   is  {\bf trivial}   if there exists an $x\in\g$ such that $\Id+t\ad_x$ is an isomorphism from $R_{t}=R+t\hat{R}$ to $R$.
\end{defi}

\begin{defi}
Let $R$ be a modified $r$-matrix on a Lie algebra $(\g, [\cdot, \cdot]_\g)$. An element $x\in\g$ is called a {\bf Nijenhuis element} associated to $R$ if $x$ satisfies
\begin{eqnarray}
\label{eq-Nij1}[[x, y]_\g, [x, z]_\g]_\g&=&0,\\
\label{eq-Nij2}[x, [x, R(y)]_\g]_\g&=&[x, R([x, y]_\g)]_\g,
\end{eqnarray}
for all $y, z\in\g$.
\end{defi}

Let $\hat{R}$ generate a trivial linear deformation of a modified $r$-matrix $R$ on a Lie algebra $(\g, [\cdot,\cdot]_\g)$. Then there exists $x\in\g$ such that
\begin{eqnarray*}
(\Id+t\ad_x)[y, z]_\g&=&[y+t[x, y]_\g, z+t[x, z]_\g]_\g,\\
 R(y+t[x, y]_\g)&=&(\Id+t\ad_x)(R(y)+t\hat{R}(y)),
\end{eqnarray*}
 for all $y, z\in\g. $ Therefore, we have $$[[x, y]_\g, [x, z]_\g]_\g=0,\quad [x, \hat{R}(y)]_\g=0,\quad R([x, y]_\g)=[x, R(y)]_\g+\hat{R}(y).$$ Thus a trivial linear deformation gives rise to a Nijenhuis element.

\begin{thm}
Let $R$ be a modified $r$-matrix on a Lie algebra $\g$. Then for any Nijenhuis element $x\in\g$, $R_t=R+t\dM_{\CE}^{R}x$ is a trivial linear deformation of the modified $r$-matrix $R$.
\end{thm}
\begin{proof}
Denote by $\hat{R}=\dM_{\CE}^{R}x$.  To show that $R_t$ is a linear deformation of $R$,  it suffices to show that \eqref{defi-defor2} and \eqref{defi-Defor} hold. Note that \eqref{defi-defor2} means that $\hat{R}$ is closed, which holds naturally since now $\hat{R}=\dM_{\CE}^{R}x$ is exact. Thus, we need to verify that Equation \eqref{defi-Defor} holds. 
For any $y, z\in\g$, by \eqref{defor-cob}, we obtain $\hat{R}(y)=[R(y), x]_\g-R([y, x]_\g)$.
Moreover, by \eqref{assmd}, \eqref{eq-Nij1} and \eqref{eq-Nij2}, it follows that
\begin{eqnarray}
\nonumber&&[R([y, x]_\g), R([z, x]_\g)]_\g\\
\nonumber&\stackrel{\eqref{assmd},\eqref{eq-Nij1}}{=}&R\Big([R([y, x]_\g), [z, x]_\g]_\g+[[y, x]_\g, R([z, x]_\g)]_\g\Big)\\
\nonumber&=&R\Big([R([y, x]_\g), [z, x]_\g]\Big)+R\Big([[y, x]_\g, R([z, x]_\g)]_\g\Big)\\
\nonumber&=&R\Big([[R([y, x]_\g), z]_\g, x]_\g\Big)+R\Big([z, [R([y, x]_\g), x]_\g]\Big)\\\nonumber&&+R\Big([[y, R([z, x]_\g)]_\g, x]_\g\Big)+R\Big([y, [x, R([z, x]_\g)]_\g]_\g\Big)\\
\nonumber&\stackrel{\eqref{eq-Nij1},\eqref{eq-Nij2}}{=}&R\Big([[R([y, x]_\g), z]_\g, x]_\g\Big)+R\Big([x, [z, [x, R(y)]_\g]_\g]\Big)\\\nonumber &&+R\Big([[y, R([z, x]_\g)]_\g, x]_\g\Big)-R\Big([x, [y, [x, R(z)]_\g]_\g]_\g\Big),
\end{eqnarray}
\begin{eqnarray*}
&&-[[R(y), x]_\g, R([z, x]_\g)]_\g\\
&=&-[[R(y), R([z, x]_\g)]_\g, x]_\g-[R(y), [x, R([z, x]_\g)]_\g]_\g\\
&\stackrel{\eqref{assmd}}{=}&-[R([R(y), [z, x]_\g]_\g), x]_\g-[R([y, R([z, x]_\g)]_\g), x]_\g\\\nonumber &&+[[y, [z, x]_\g]_\g, x]_\g-[R(y), [x, R([z, x]_\g)]_\g]_\g\\
&\stackrel{\eqref{eq-Nij1},\eqref{eq-Nij2}}{=}&-[x, [x, R([R(y), z]_\g)]_\g]_\g-[R([z, [R(y), x]_\g]_\g), x]_\g-[R([y, R([z, x]_\g)]_\g), x]_\g\\\nonumber &&+[[y, [z, x]_\g]_\g, x]_\g+[x, [x, [R(y), R(z)]_\g]_\g]_\g+[[x, [R(y), x]_\g]_\g, R(z)]_\g\\
&\stackrel{\eqref{assmd},\eqref{eq-Nij2}}{=}&[x, [x, R([y, R(z)]_\g)]_\g]_\g-[x, [x, [y, z]_\g]_\g]_\g+[[y, [z, x]_\g]_\g, x]_\g\\
&&-[R([z, [R(y), x]_\g]_\g), x]_\g-[R([y, R([z, x]_\g)]_\g), x]_\g-[[x, R([x, y]_\g]_\g), R(z)]_\g
\end{eqnarray*}
and
\begin{eqnarray*}
&&-[R([y, x]_\g), [R(z), x]_\g]_\g\\
&\stackrel{\eqref{assmd}}{=}&[R([[R(z), y]_\g, x]_\g), x]_\g+[R([y, [R(z), x]_\g]_\g), x]_\g+[R([z, R([y, x]_\g)]_\g), x]_\g\\
&&-[[z, [y, x]_\g]_\g, x]_\g+[R(z), [x, R([y, x]_\g)]_\g]_\g.
\end{eqnarray*}

By \eqref{eq-Nij1} and above equations, we have
\begin{eqnarray*}
&&[\hat{R}(y), \hat{R}(z)]_\g-\hat{R}([\hat{R}(y), z]_\g+[y, \hat{R}(z)]_\g)\\
&=&[[R(y), x]_\g-R([y, x]_\g), [R(z), x]_\g-R([z, x]_\g)]_\g\\
&&-[R([[R(y), x]_\g-R([y, x]_\g),z]), x]_\g+R([[[R(y), x]_\g-R([y, x]_\g), z]_\g, x]_\g)\\
&&-[R([y, [R(z), x]_\g-R([z, x]_\g)]_\g), x]_\g+R([[y, [R(z), x]_\g-R([z, x]_\g)]_\g, x]_\g)\\
&=&-[[R(y), x]_\g, R([z, x]_\g)]_\g+[R([y, x]_\g), R([z, x]_\g)]_\g-[R([y, x]_\g), [R(z), x]_\g]_\g\\
&&-[R([[R(y), x]_\g-R([y, x]_\g),z]), x]_\g+R([[[R(y), x]_\g-R([y, x]_\g), z]_\g, x]_\g)\\
&&-[R([y, [R(z), x]_\g-R([z, x]_\g)]_\g), x]_\g+R([[y, [R(z), x]_\g-R([z, x]_\g)]_\g, x]_\g)\\
&=&0.
\end{eqnarray*}
Thus $R_t=R+t\dM_{\CE}^{R}x$ is a linear deformation of the modified $r$-matrix $R$. Since $x\in\g$ is a Nijenhuis element, we have $(\Id+t\ad_x)[y, z]_\g=[y+t[x, y]_\g, z+t[x, z]_\g]_\g$ and $R\circ(\Id+t\ad_x)=(\Id+t\ad_x)\circ(R+t\dM_{\CE}^{R}x)$. Thus for any Nijenhuis element $x\in\g$, $R_t=R+t\dM_{\CE}^{R}x$ is a trivial linear deformation of the modified $r$-matrix $R$.
\end{proof}

At the end of this section, we consider the relation between linear deformations of modified $r$-matrices and  linear deformations of the induced Lie algebras. Recall that a skew-symmetric bilinear map $\omega: \wedge^{2}\g\lon\g$ generates a linear deformation of a Lie algebra $(\g, [\cdot,\cdot]_{\g})$ if $[\cdot, \cdot]_t=[\cdot, \cdot]_\g+t\omega$ defines a Lie algebra structure on $\g$ for all $t\in (-\epsilon, \epsilon)$.

\begin{pro}
Let $\hat{R}$ generate a linear deformation of a modified $r$-matrix $R$ on a Lie algebra $(\g, [\cdot,\cdot]_\g)$. Then $\omega$ defined by
\begin{equation*}
\omega(x, y)=[\hat{R}(x),y]_\g+[x,\hat{R}(y)]_\g,\quad \forall x,y\in\g,
\end{equation*}
generates a linear deformation of the Lie algebra $(\g,[\cdot,\cdot]_R)$ given by the modified $r$-matrix $R$, which is exactly the one associated to the linear deformation of the modified $r$-matrix $R$.
\end{pro}
\begin{proof}
It is obvious that $$[x, y]_{R_t}=[R(x), y]_\g+[x, R(y)]_\g+t([\hat{R}(x),y]_\g+[x,\hat{R}(y)]_\g)=[x, y]_R+t\omega(x, y).$$ Since $[\cdot, \cdot]_{R_t}$ are Lie algebra structures, we have that $\omega$ generates a linear deformation of the Lie algebra $(\g,[\cdot,\cdot]_R)$ given by the modified $r$-matrix $R$.
\end{proof}
The notion of a Nijenhuis operator on a Lie algebra $(\g, [\cdot, \cdot]_\g)$ was given in \cite{Dor}, which gives rise to a trivial linear deformation of the Lie algebra $(\g, [\cdot, \cdot]_\g)$.

\begin{defi}\rm(\cite{Dor})
Let $(\g, [\cdot, \cdot]_\g)$ be a Lie algebra. A linear map $N:\g\lon\g$ is called {\bf Nijenhuis operator} if
\begin{equation*}
[N(x), N(y)]_\g=N([N(x), y]_\g+[x, N(y)]_\g)-N^{2}([x, y]_\g), \quad \forall x, y\in\g.
\end{equation*}
\end{defi}

\begin{thm}
Let $x\in\g$ be a Nijenhuis element associated to a modified matrix $R$. Then $\ad_x$ is a Nijenhuis operator on the Lie algebra $(\g, [\cdot, \cdot]_R)$.
\end{thm}
\begin{proof}
For any $x, y, z\in\g$, by \eqref{eq-Nij1} and \eqref{eq-Nij2}, we have
\begin{eqnarray*}
&&[\ad_xy, \ad_xz]_{R}\\
&=&[R([x, y]_\g), [x, z]_\g]_\g+[[x, y]_\g, R([x, z]_\g)]_\g\\
&=&[[R([x, y]_\g), x]_\g, z]_\g+[x, [R([x, y]_\g),z]_\g]_\g+[[x, R([x, z]_\g)]_\g, y]_\g+[x, [y, R([x, z]_\g)]_\g]_\g\\
&=&-[[x, [x, R(y)]_\g]_\g, z]_\g+[x, [R([x, y]_\g),z]_\g]_\g+[[x, [x, R(z)]_\g]_\g, y]_\g+[x, [y, R([x, z]_\g)]_\g]_\g\\
&=&-[x, [[x, R(y)]_\g, z]_\g]_\g+[x, [R([x, y]_\g),z]_\g]_\g+[x, [[x, R(z)]_\g, y]_\g]_\g+[x, [y, R([x, z]_\g)]_\g]_\g
\end{eqnarray*}
and
\begin{eqnarray*}
&&\ad_x([\ad_xy, z]_R+[y, \ad_xz]_R)-\ad_x^{2}([y, z]_R)\\
&=&[x, [R([x, y]_\g), z]_\g]_\g+[[x, y]_\g, R(z)]_\g+[x, [R(y), [x, z]_\g]_\g]_\g\\
&&[x, [y, R([x, z]_\g)]_\g]_\g-[x, [x, [R(y), z]_\g]_\g]_\g-[x, [x, [y, R(z)]_\g]_\g]_\g.
\end{eqnarray*}
Thus $[\ad_xy, \ad_xz]_{R}=\ad_x([\ad_xy, z]_R+[y, \ad_xz]_R)-\ad_x^{2}([y, z]_R)$, which implies that $\ad_x$ is a Nijenhuis operator on the Lie algebra $(\g, [\cdot, \cdot]_R)$.
\end{proof}

\section{Applications}\label{sec:app}

In this section, we give some applications of the above deformation theories, including deformations of complement of the diagonal Lie algebra $\g_\Delta$ and compatible Poisson structures.

\subsection{Deformations of complements}

Let $(\g, [\cdot,\cdot]_{\g})$ be a Lie algebra, then we have a direct-product Lie algebra structure $[\cdot,\cdot]_{\oplus}$ on $\g\oplus\g$, where
\begin{equation*}
[(x_1,y_1),(x_2,y_2)]_{\oplus}=([x_1,x_2]_{\g},[y_1,y_2]_{\g}), \quad\forall x_i, y_i\in\g, i=1,2.
\end{equation*}
Define the subspace $\g_{\Delta}$ by $\g_{\Delta}=\{(x,x)|\forall x\in\g\}$ and the subspace $\g_{-\Delta}=\{(x,-x)|\forall x\in\g\}$. It is obvious that $\g_\Delta$ is a Lie subalgebra of $\g\oplus\g$, while $\g_{-\Delta}$ is not a Lie subalgebra. To find a complement of $\g_\Delta$ which is also a Lie subalgebra, it is natural to consider the graph of certain linear map from $\g_{-\Delta}$ to $\g_{\Delta}$. It is known that a complement of $\g_\Delta$ is isomorphic to a graph of a linear map from $\g_{-\Delta}$ to $\g_\Delta$. 
Let $R\in\gl(\g)$ be a linear map.  Define a linear map $\hat{R}: \g_{-\Delta}\lon \g_{\Delta}$ by
\begin{equation*}
\hat{R}(x,-x)=(-R(x),-R(x)), \quad \forall x\in\g.
\end{equation*}

\begin{pro}\label{prographR}
With the above notations, the graph  $\huaG(\hat{R}):=\{\hat{R}u+u|u\in\g_{-\Delta}\}$ is a Lie subalgebra of $(\g\oplus\g, [\cdot,\cdot]_{\oplus})$ if and only if $R$ is a modified $r$-matrix.
\end{pro}
\begin{proof}
 For all $x,y\in\g$, we have
\begin{eqnarray*}
&&[(-R(x),-R(x))+(x,-x), (-R(y),-R(y))+(y,-y)]_{\oplus}\\
&=&([x,y]_{\g},[x,y]_{\g})+([R(x),R(y)]_{\g},[R(x),R(y)]_{\g})\\
&&+(-[R(x),y]_{\g},[R(x),y]_{\g})+(-[x,R(y)]_{\g},[x,R(y)]_{\g})\\
&=&([x,y]_{\g}+[R(x),R(y)]_{\g},[x,y]_{\g}+[R(x),R(y)]_{\g})\\
&&+(-[R(x),y]_{\g}-[x,R(y)]_{\g}, [R(x),y]_{\g}+[x,R(y)]_{\g}).
\end{eqnarray*}
Thus $\huaG(\hat{R})$ is a Lie subalgebra if and only if
\begin{equation*}
R([R(x),y]_{\g}+[x,R(y)]_{\g})=[x,y]_{\g}+[R(x),R(y)]_{\g},
\end{equation*}
i.e. $R$ is a modified $r$-matrix.
\end{proof}

\begin{pro}
Let $R$ be a modified $r$-matrix. Then $(\g_\Delta, \huaG(\hat{R}))$ is a matched pair of Lie algebras.
\end{pro}
\begin{proof}

It is obvious that $\g\oplus\g=\g_{\Delta}\oplus\huaG(\hat{R})$ since $\g_{\Delta}\bigcap\huaG(\hat{R})=0$.  Then the conclusion follows from the fact that both $\g_{\Delta}$ and $\huaG(\hat{R})$ are Lie subalgebras.
\end{proof}

Summarizing the above studies, we have the following conclusion.

\begin{thm}
  Let $R_t$ be a geometric deformation of a modified $r$-matrix $R$. Then $ \huaG(\hat{R_t})$ is a deformation of the complement $ \huaG(\hat{R})$. Moreover,  $(\g_\Delta, \huaG(\hat{R_t}))$ are  matched pairs of Lie algebras.
\end{thm}

\subsection{Compatible Poisson structures} A {\bf compatible Poisson structure} consists of two Poisson structures $\pi, \pi'$ on a manifold $M$ such that $\pi+\pi'$ is also a Poisson structure on the manifold $M$.

Let $R$ be a modified $r$-matrix on a Lie algebra $\g$. Then $(\g, [\cdot, \cdot]_R)$ is a Lie algebra and we denote by $(\g^*,\{\cdot, \cdot\}_R)$ the corresponding linear Poisson manifold.
\begin{pro}
Let $R$ be a modified $r$-matrix on a Lie algebra $\g$ and $R_t=R+t\hat{R}$ be a linear deformation of $R$. For any $t_1, t_2\in \mathbb{R}$, $\{\cdot, \cdot\}_{R_{t_1}}$ and $\{\cdot, \cdot\}_{R_{t_2}}$ are compatible Poisson structures on $\g^*$.
\end{pro}
\begin{proof}
By the fact that $R+\frac{t_1+t_2}{2}\hat{R}$ is also a modified $r$-matrix on the Lie algebra $\g$, we have
\begin{equation*}
[x, y]_{R_{t_1}}+[x, y]_{R_{t_2}}=2([R(x)+\frac{t_1+t_2}{2}\hat{R}(x), y]_\g+[x, R(y)+\frac{t_1+t_2}{2}\hat{R}(y)]),
\end{equation*}
which implies that $[\cdot, \cdot]_{R_{t_1}}+[\cdot, \cdot]_{R_{t_2}}$ is also a Lie bracket on the Lie algebra $\g$ by \eqref{defi-lie}. Thus, for any $t_1, t_2\in \mathbb{R}$, $\{\cdot, \cdot\}_{R_{t_1}}$ and $\{\cdot, \cdot\}_{R_{t_2}}$ are compatible Poisson structures on $\g^*$.
\end{proof}

\vspace{2mm}
\noindent
{\bf Acknowledgements. } We give warmest thanks to Jianghua Lu and Chenchang Zhu for helpful comments. This research is supported by NSFC (11922110).

On behalf of all authors, the corresponding author states that there is no conflict of interest.


\begin{thebibliography}{a}

\bibitem{BGN} C. Bai, L. Guo and X. Ni, Nonabelian generalized Lax pairs, the classical Yang-Baxter equation and PostLie algebras. {\em Comm. Math. Phys.} {\bf 297} (2010) 553-596.

\bibitem{Bor}
M. Bordemann, Generalized Lax Pairs, the modified classical Yang-Baxter equation, and affine geometry of Lie groups. \emph{Comm. Math. Phys.} {\bf 135} (1990), 201-216.

\bibitem{Borisov}
D. V. Borisov, Formal deformations of morphisms of associative algebras. {\em Int. Math. Res. Not.} {\bf 41} (2005), 2499-2523.

\bibitem{Ch-Ei}
C. Chevalley and S. Eilenberg,
\newblock Cohomology theory of {L}ie groups and {L}ie algebras.
\newblock {\em Trans. Amer. Math. Soc.} {\bf 63} (1948), 85-124.

\bibitem{CSS}
M. Crainic, F. Schatz and I. Struchiner, A survey on stability and rigidity results for Lie algebras. \emph{Indag. Math. (N.S.)} {\bf 25} (2014), no. 5, 957-976.

\bibitem{Das}
A. Das, Deformations of associative Rota-Baxter operators. \emph{ J. Algebra} {\bf560} (2020), 144-180.

\bibitem{Dor}
 I. Dorfman, Dirac Structures and Integrability of Nonlinear Evolution Equations. \emph{Wiley Chichester} 1993.

\bibitem{EF}
K. Ebrahimi-Fard, Loday-type algebras and the Rota-Baxter relation. \emph{Lett. Math. Phys.} {\bf 61} (2002), no. 2, 139-147.

\bibitem{FKK} K. Ebrahimi-Fard, I. Mencattini and H. Munthe-Kaas,  Post-Lie algebras and factorization theorems. \emph{J. Geom. Phys.} {\bf 119} (2017), 19-33.

\bibitem{FM}K. Ebrahimi-Fard and I. Mencattini,   Post-Lie algebras, factorization theorems and isospectral flows. Discrete mechanics, geometric integration and Lie-Butcher series, 231-285, \emph{Springer Proc. Math. Stat.}, 267, Springer, Cham, 2018.

    \bibitem{Fregier-Zambon-1}
Y. Fr\'egier and M. Zambon, Simultaneous deformations and Poisson geometry. \emph{Compos. Math.} {\bf 151} (2015), 1763-1790.

\bibitem{Fregier-Zambon-2}
Y. Fr\'egier and M. Zambon, Simultaneous deformations of algebras and morphisms via derived brackets. \emph{J. Pure Appl. Algebra} {\bf 219 } (2015), 5344-5362.

    \bibitem{Ge0}
M. Gerstenhaber, The cohomology structure of an associative ring. \emph{Ann. Math.} {\bf 78} (1963), 267-288.

\bibitem{Ge}
M. Gerstenhaber, On the deformation of rings and algebras. \emph{Ann. Math. (2) } {\bf 79} (1964), 59-103.

    \bibitem{GLS}
L. Guo, H. Lang and Y. Sheng, Integration and geometrization of Rota-Baxter Lie algebras.  {\em Adv. Math.}  {\bf387} (2021), Paper No. 107834, 34 pp.

\bibitem{JSZ}
J. Jiang, Y. Sheng and C. Zhu, Lie theory and cohomology of relative Rota-Baxter operators. arXiv:2108.02627.

\bibitem{Ku}
B. A. Kupershmidt, What a classical $r$-matrix really is. \emph{J. Nonlinear Math. Phys.} {\bf 6} (1999), 448-488.


\bibitem{Li}
L. C. Li, Classical $r$-matrices and compatible Poisson structures for Lax equations on Poisson algebras. \emph{Comm. Math. Phys.} {\bf 203} (1999), no. 3, 573-592.

\bibitem{Lu} J. Lurie, \emph{ DAG X: Formal moduli problems}, available at http://www.math.harvard.edu/~lurie/papers/DAG-X.pdf

\bibitem{NR} A. Nijenhuis  and R. Richardson,  Cohomology and deformations in graded Lie algebras. {\em Bull.
Amer. Math. Soc.} {\bf 72} (1966), 1-29.

\bibitem{NR2} A. Nijenhuis and R. Richardson,  Commutative  algebra cohomology and deformations of Lie and associative algebras. {\em J. Algebra} {\bf 9} (1968), 42-105.

    \bibitem{Pr} J. P. Pridham, Unifying derived deformation theories. \emph{Adv. Math.} {\bf 224} (2010), 772-826.


\bibitem{RS1} A.~G. Reyman and M.~A. Semenov-Tian-Shansky, Reduction of Hamilton systems, affine Lie algebras and Lax equations. \emph{Invent. Math.} {\bf 54} (1979), 81-100.

\bibitem{RS88} N. Yu. Reshetikhin and M. A. Semenov-Tian-Shansky, Quantum R-matrices and factorization problems. \emph{J. Geom. Phys.} {\bf5} (1988), no. 4, 533-550.

\bibitem{STS}
M. Semonov-Tian-Shansky, What is a classical $r$-matrix? \emph{Funct. Anal. Appl.} {\bf 17} (1983), 259-272.

\bibitem{STS2} M. A. Semenov-Tian-Shansky, Integrable systems and factorization problems. {\em Operator Theory: Advances and Applications} {\bf 141} (2003), 155-218, Birkh\"auser Verlag Basel.

\bibitem{Sza}B. Szablikowski,   Classical $r$-matrix like approach to Frobenius manifolds, WDVV equations and flat metrics. \emph{J. Phys. A} {\bf48} (2015), no. 31, 315203, 47 pp.

\bibitem{SzaB}
B. Szablikowski and M. Blaszak, On deformations of standard R-matrices for integrable infinite-dimensional systems. \emph{J. Math. Phys.} {\bf 46} (2005), no. 4, 042702, 12 pp.

\bibitem{TBGS}
R. Tang, C. Bai, L. Guo and Y. Sheng, Deformations and their controlling cohomologies of $\huaO$-operators. \emph{Comm. Math. Phys.} \textbf{368} (2019), 665-700.



\bibitem{ZGG1}
X. Zhang, X. Gao and L. Guo, Modified Rota-Baxter algebras, shuffle products and Hopf algebras. \emph{Bull. Malays. Math. Sci. Soc.} {\bf 42} (2019), no. 6, 3047-3072.

\bibitem{ZGG2}
X. Zhang, X. Gao and L. Guo, Free modified Rota-Baxter algebras and Hopf algebras. \emph{Int. Electron. J. Algebra} {\bf 25} (2019), 12-34.

\end{thebibliography}
 \end{document}